\documentclass[12pt,fleqn]{article}

\usepackage{amsmath,amssymb,amsthm}
\usepackage[margin=1in]{geometry}
\usepackage{xcolor}
\usepackage{url}
\usepackage{float}
\usepackage{array}
\usepackage{comment}
\usepackage{booktabs}
\usepackage[hypertexnames=false,bookmarksnumbered=true,final]{hyperref}
\usepackage{algorithm,algpseudocode}
\usepackage[capitalize,sort]{cleveref}

\def\colorschemesepia{sepia}
\def\colorschemedark{dark}
\def\colorschemelight{light}

\ifx\colorscheme\undefined
\let\colorscheme\colorschemelight
\fi

\ifx\colorscheme\colorschemelight
\colorlet{textColor}{black}
\colorlet{bgColor}{white}
\fi

\ifx\colorscheme\colorschemesepia
\definecolor{textColor}{HTML}{433423}
\definecolor{bgColor}{HTML}{fbf0da}
\fi

\ifx\colorscheme\colorschemedark
\definecolor{textColor}{HTML}{bdc1c6}
\definecolor{bgColor}{HTML}{202124}
\definecolor{textBlue}{HTML}{8ab4f8}
\definecolor{textRed}{HTML}{f9968b}
\definecolor{textGreen}{HTML}{81e681}
\definecolor{textPurple}{HTML}{c58af9}
\else
\colorlet{textBlue}{blue!50!black}
\colorlet{textRed}{red!50!black}
\colorlet{textGreen}{green!50!black}
\definecolor{textPurple}{HTML}{681da8}
\fi

\ifx\colorscheme\colorschemelight\else
\pagecolor{bgColor}
\color{textColor}
\fi

\usepackage{iftex}
\usepackage{url}
\usepackage{mathtools}
\usepackage{algorithm,algpseudocode}
\usepackage{comment}

\def\colorschemesepia{sepia}
\def\colorschemedark{dark}
\def\colorschemelight{light}

\ifx\colorscheme\undefined
\let\colorscheme\colorschemelight
\fi

\ifx\colorscheme\colorschemelight
\colorlet{textColor}{black}
\colorlet{bgColor}{white}
\fi

\ifx\colorscheme\colorschemesepia
\definecolor{textColor}{HTML}{433423}
\definecolor{bgColor}{HTML}{fbf0da}
\fi

\ifx\colorscheme\colorschemedark
\definecolor{textColor}{HTML}{bdc1c6}
\definecolor{bgColor}{HTML}{202124}
\definecolor{textBlue}{HTML}{8ab4f8}
\definecolor{textRed}{HTML}{f9968b}
\definecolor{textGreen}{HTML}{81e681}
\definecolor{textPurple}{HTML}{c58af9}
\else
\colorlet{textBlue}{blue!50!black}
\colorlet{textRed}{red!50!black}
\colorlet{textGreen}{green!50!black}
\definecolor{textPurple}{HTML}{681da8}
\fi

\ifx\colorscheme\colorschemelight\else
\pagecolor{bgColor}
\color{textColor}
\fi

\hypersetup{
colorlinks,
linkcolor=textRed,
citecolor=textRed,
urlcolor=textBlue
}

\makeatletter
\g@addto@macro{\UrlBreaks}{%
\do\/%
\do\a\do\b\do\c\do\d\do\e\do\f\do\g\do\h\do\i\do\j\do\k\do\l\do\m%
\do\n\do\o\do\p\do\q\do\r\do\s\do\t\do\u\do\v\do\w\do\x\do\y\do\z%
\do\A\do\B\do\C\do\D\do\E\do\F\do\G\do\H\do\I\do\J\do\K\do\L\do\M%
\do\N\do\O\do\P\do\Q\do\R\do\S\do\T\do\U\do\V\do\W\do\X\do\Y\do\Z%
\do\0\do\1\do\2\do\3\do\4\do\5\do\6\do\7\do\8\do\9%
}
\makeatother

\allowdisplaybreaks
\newcolumntype{L}{>{$\displaystyle}l<{$}}
\algnewcommand{\LineComment}[1]{\State \textcolor{gray}{\texttt{//} \textit{#1}}}
\algdef{SE}[REPEATN]{RepeatN}{End}[1]{\algorithmicrepeat\ #1 \textbf{times}}{\algorithmicend}

\makeatletter
\@ifpackageloaded{enumitem}{%
\newenvironment*{tightemize}{\begin{itemize}[noitemsep]}{\end{itemize}}%
\newenvironment*{tightenum}{\begin{enumerate}[noitemsep]}{\end{enumerate}}%
}{%
}
\makeatother

\newtheorem{theorem}{Theorem}
\newtheorem{definition}{Definition}

\newtheorem{lemma}[theorem]{Lemma}

\newtheorem{observation}[theorem]{Observation}

\crefname{claim}{Claim}{Claims}
\crefname{property}{Property}{Properties}
\crefname{example}{Example}{Examples}
\crefname{observation}{Observation}{Observations}
\crefname{transformation}{Transformation}{Transformations}

\usepackage{amsmath}

\newcommand*{\thmdep}[2]{}

\newcommand*{\Th}{^{\textrm{th}}}

\newcommand*{\wLoG}{without loss of generality}

\let\eps\varepsilon

\newcommand*{\ihat}{\widehat{\imath}}
\newcommand*{\jhat}{\widehat{\jmath}}
\newcommand*{\vhat}{\widehat{v}}

\newcommand*{\Ahat}{\widehat{A}}
\newcommand*{\Bhat}{\widehat{B}}
\newcommand*{\Ical}{\mathcal{I}}
\newcommand*{\Icalhat}{\widehat{\mathcal{I}}}

\newcommand*{\defeq}{:=}

\DeclareMathOperator*{\E}{E}

\DeclareMathOperator*{\argmax}{argmax}

\newcommand*{\PROP}{\mathrm{PROP}}
\newcommand*{\MMS}{\mathrm{MMS}}
\DeclareMathOperator{\ECE}{\mathtt{ECE}}
\DeclareMathOperator{\resolveCycle}{\mathtt{resolveCycle}}
\DeclareMathOperator{\toOrdered}{\mathtt{toOrdered}}
\DeclareMathOperator{\pickBySeq}{\mathtt{pickBySeq}}
\DeclareMathOperator{\ECEGTwo}{\mathtt{ECEG2}}
\DeclareMathOperator{\ECECTwo}{\mathtt{ECEC2}}

\let\shortcite\cite
\let\citet\cite
\let\citep\cite

\title{Best-of-Both-Worlds Fairness of the Envy-Cycle-Elimination Algorithm}
\author{
Jugal Garg%
\thanks{Department of Industrial \& Enterprise Engineering, University of Illinois at Urbana-Champaign, USA}
\\ \texttt{\small jugal@illinois.edu}
\and
Eklavya Sharma\footnotemark[1]
\\ \texttt{\small eklavya2@illinois.edu}
}
\date{\empty}

\begin{document}

\maketitle

\begin{abstract}
We consider the problem of fairly dividing indivisible goods among agents with additive valuations. It is known that an Epistemic EFX and $2/3$-MMS allocation can be obtained using the Envy-Cycle-Elimination (ECE) algorithm\texorpdfstring{ \cite{caragiannis2023new,barman2020approximation}}{}.
In this work, we explore whether this algorithm can be randomized to also ensure ex-ante proportionality.
For two agents, we show that a randomized variant of ECE can compute an ex-post EFX and ex-ante envy-free allocation in near-linear time. However, for three agents, we show that several natural randomization methods for ECE fail to achieve
 ex-ante proportionality.

\end{abstract}

\section{Introduction}
\label{sec:intro}

In the problem of fair division, a set $M$ of $m$ indivisible goods
must be allocated to $n$ agents in a \emph{fair} manner.
Intuitively, \emph{fairness} means that we shouldn't favor certain agents over others.
However, stating a precise mathematical definition of fairness
and developing algorithms for fairly allocating the goods has turned out to be
a highly non-trivial problem, and has led to extensive research for the last fifteen years
\citep{brams1996fair,moulin2004fair}.

Formally, each agent $i$ has a valuation function $v_i$ that takes as input a subset of goods
and returns a number representing how much the agent values that subset.
Our goal is to find an allocation, i.e., a tuple $(A_1, \ldots, A_n)$
where $A_i$ is the set of goods that each agent $i$ receives,
such that the allocation is fair.

Many notions of fairness have been proposed in the literature.
One of the first notions to be studied is \emph{envy-freeness} \citep{foley1966resource}.
An allocation of items is envy-free (EF) if each agent prefers her own bundle to any other agent's bundle.
Formally, allocation $A$ is envy-free if $v_i(A_i) \ge v_i(A_j)$ for all $i \neq j$.
Another popular fairness notion is \emph{proportionality}.
An allocation is proportionally fair to an agent if the value of her own bundle is at least
$1/n$ times her value for the entire set of goods.
Formally, allocation $A$ is proportional (PROP) if $v_i(A_i) \ge v_i(M)/n$ for every agent $i$.
It's easy to see that envy-freeness and proportionality cannot always be guaranteed.
E.g., if there is a single good and two agents, some agent gets nothing.
Hence, relaxations of envy-freeness and proportionality are studied.

The notion `envy-free up to any good', abbreviated as \emph{EFX},
is widely regarded as one of the best relaxations of envy-freeness.
In an allocation, agent $i$ is said to \emph{strongly envy} another agent $j$ if
there is a good $g$ in $j$'s bundle such that $i$ envies $j$ even after $g$ is removed from $j$'s bundle.
An allocation is called EFX if no agent strongly envies any other agent.
However, despite significant effort by researchers
\citep{plaut2020almost,caragiannis2019envy,chaudhury2020efx,amanatidis2021maximum,chaudhury2021little,mahara2021extension,berger2021almost},
the guaranteed existence of EFX allocations has remained an open problem
ever since the notion was proposed in 2016 by \citet{caragiannis2016unreasonable,caragiannis2019unreasonable}.
Hence, EFX has been further relaxed to other notions, like EF1 \citep{budish2011combinatorial,lipton2004approximately},
approximate EFX \citep{plaut2020almost}, and Epistemic EFX \citep{caragiannis2023new}.

The most-widely studied relaxation of proportionality is \emph{maximin share}.
The maximin share of an agent $i$ is the maximum value she can receive by partitioning
the goods into $n$ bundles and picking the least valued bundle, i.e.,
\[ \MMS_i \defeq \max_X \min_{j=1}^n v_i(X_j). \]
We would like to give each agent a bundle that she values at least as much as her maximin share.
\citet{kurokawa2018fair,feige2022tight} showed that this is impossible,
so research turned towards finding allocations that guarantee each agent
a large multiplicative factor of their maximin share
\citep{kurokawa2018fair,barman2020approximation,ghodsi2018fair,garg2021improved,akrami2024breaking}.

Many algorithms for fair division have been proposed.
However, one algorithm (and variations thereof) features prominently in many works:
the Envy-Cycle-Elimination algorithm, hereafter abbreviated as ECE.
It was first used by \citet{lipton2004approximately} to get EF1 allocations,
and thereafter has been repurposed for several other fairness notions
\citep{barman2020approximation,amanatidis2020multiple,chaudhury2021little,bhaskar2021approximate,caragiannis2023new,chaudhury2024efx,berger2022almost,feldman2023breakingArxiv}.
In this work, we investigate the implications of randomizing ECE.

Randomized fair division is a line of work started by \cite{aziz2023best}.
For example, when there is a single good, giving the good to an agent selected uniformly randomly is fair.
In expectation, both agents get half of the good, so they are, in a sense, envy-free
(we formally define this in \cref{sec:prelims:bobw}).
In this setting, fairness is defined as a property of distributions of allocations.
Thus, randomization can help us get very strong fairness guarantees, like envy-freeness,
which are impossible to achieve otherwise.

However, randomness by itself is unsatisfactory: giving all the goods to a random agent is fair in expectation because each agent has equal opportunity, but is unfair after the randomness is realized due to the large disparity in this allocation.
To fix this, one can aim for a \emph{best-of-both-worlds} approach
\citep{babaioff2022best,aziz2023best,feldman2023breakingArxiv,AleksandrovAGW15},
i.e., find a distribution of allocations where each allocation in the support of the distribution
is fair (ex-post fairness), and the entire distribution is fair in a randomized sense (ex-ante fairness).
For example, \citet{babaioff2022best} gave an algorithm whose output is ex-post $1/2$-MMS and ex-ante proportional,
i.e., it outputs a distribution over $1/2$-MMS allocations, such that
each agent's expected value of her bundle is at least $v_i(M)/n$.

\citet{feldman2023breakingArxiv} gave a randomized algorithm whose output is
ex-ante $1/2$-EF (i.e., $\E(v_i(A_i)) \ge (1/2)\E(v_i(A_j))$ for all $j \neq i$,
where $A$ is the random allocation output by their algorithm),
ex-post EF1, and ex-post $1/2$-EFX. They did this by randomizing ECE.
To provide some intuition behind how they did this, let's first get an overview of how ECE works.
ECE starts with an empty allocation (i.e., no agent has any goods),
and repeately performs one of these two operations:
\begin{enumerate}
\item Give an unallocated good to a \emph{deserving} agent.
\item Identify an \emph{envy cycle}, i.e., a cycle of agents where each agent
    envies the next one in the cycle, and then shuffle the bundles among these agents.
\end{enumerate}
There may be multiple deserving agents or multiple envy cycles to choose from.
However, it is known \citep{lipton2004approximately,plaut2020almost} that ECE's output is
EF1 and $1/2$-EFX, regardless of how we pick the deserving agent or the envy cycle.
\citet{feldman2023breakingArxiv} showed that, by picking the envy cycle randomly
from a carefully selected distribution, one can additionally get ex-ante $1/2$-EF.

Appendix D of \cite{babaioff2022bestArxiv} randomizes ECE differently to get
ex-ante $1/2$-PROP (i.e., $\E(v_i(A_i)) \ge v_i(M)/(2n)$),
ex-post EF1, and ex-post $1/2$-MMS.

\subsection{Our Contribution}

We investigate whether ECE can be randomized to get better fairness guarantees.

We first start with the special case of two agents.
\citet{feldman2023breakingArxiv} get ex-post EFX and ex-ante EF for two agents
using the cut-and-choose protocol. However, their algorithm runs in exponential time.
In \cref{sec:n2algo}, we show how to randomize ECE to get
ex-ante envy-freeness, ex-post EFX, and ex-post $4/5$-MMS.
Our algorithm runs in $O(m\log m)$ time.%
\footnote{Independently, \cite{bu2024best} also obtained a polynomial-time algorithm
for ex-ante envy-freeness and ex-post EFX for two agents.
However, their algorithm uses very different techniques (local search).}
We also extend our result to the setting of fair division of chores
and get ex-ante envy-freeness, ex-post EFX, and ex-post $7/6$-MMS in $O(m \log m)$ time.

Next, we study the setting with more than two agents.
\citet{barman2020approximation} showed that the envy-cycle-elimination algorithm,
when combined with an ordering trick from \citet{bouveret2016characterizing},
outputs a $2/3$-MMS allocation. \citet{caragiannis2023new} showed that
this algorithm's output is also Epistemic EFX.
We believe these to be more satisfactory ex-post guarantees compared to
\citet{feldman2023breakingArxiv} ($1/2$-EFX and EF1)
and \citet{babaioff2022bestArxiv} ($1/2$-MMS and EF1).
Hence, a natural line of investigation is to see if this algorithm can be
randomized to get ex-ante proportionality.

We fell short of this goal and could not find such a randomization.
However, in \cref{sec:ce}, we showed that for three agents, no \emph{natural} randomization
of ECE can be used to get ex-ante proportionality
(we formally define \emph{natural} in \cref{sec:ce}).

\section{Preliminaries}
\label{sec:prelims}

\subsection{Fair Division Instances and Allocations}
\label{sec:prelims:basics}

A fair division instance is represented by a tuple $(N, M, (v_i)_{i \in N})$,
where $N$ is the set of agents, $M$ is the set of items,
and $v_i$ is agent $i$'s valuation function, i.e.,
$v_i(j) \in \mathbb{R}$ is agent $i$'s value for item $j \in M$.
We let valuation functions be additive, i.e., for a subset $S \subseteq M$ of items,
define $v_i(S) \defeq \sum_{j \in M} v_i(j)$.

When $v_i(j) \ge 0$ for all $i \in N$ and $j \in M$,
the items are called \emph{goods}.
When $v_i(j) \le 0$ for all $i \in N$ and $j \in M$,
the items are called \emph{chores}.

For any integer $t \ge 0$, let $[t] \defeq \{1, 2, \ldots, t\}$.
Let $n \defeq |N|$ and $m \defeq |M|$.
Often, we can assume \wLoG{} that $N = [n]$ and $M = [m]$.

An allocation is a sequence $A \defeq (A_i)_{i \in N}$, where $A_i$,
called agent $i$'s \emph{bundle}, denotes the set of items that agent $i$ gets.
Hence, $A_i \cap A_j = \emptyset$ for any two agents $i$ and $j$, and $\bigcup_{i \in N} A_i = M$.

\subsection{Notions of Fairness}
\label{sec:prelims:fnotions}

Many different notions of fairness are known.
We begin with two of the simplest notions, \emph{envy-freeness} and \emph{proportionality}.

\begin{definition}[envy-freeness]
For a fair division instance $\Ical \defeq (N, M, (v_i)_{i \in N})$,
agent $i$ is said to \emph{envy} agent $j$ in allocation $A$ if $v_i(A_i) < v_i(A_j)$.
Agent $i$ is \emph{envy-free} (EF) in $A$ if $i$ doesn't envy any other agent.
$A$ is EF if every agent is EF in $A$.
\end{definition}

\begin{definition}[proportionality]
For a fair division instance $\Ical \defeq (N, M, (v_i)_{i \in N})$,
agent $i$'s \emph{proportional share}, denoted as $\PROP_i$, equals $v_i(M)/|N|$.
Agent $i$ is \emph{PROP-satisfied} by allocation $A$ if $v_i(A_i) \ge \PROP_i$.
$A$ is \emph{proportional} (PROP) if every agent is PROP-satisfied by $A$.
\end{definition}

Envy-free and proportional allocations don't always exist
(e.g., if there are 2 agents and a single item).
Hence, relaxations of these notions have been studied.
A well-known relaxation of envy-freeness is EFX \cite{caragiannis2019unreasonable}.

\begin{definition}[EFX]
Let $\Ical \defeq (N, M, (v_i)_{i \in N})$ be a fair division instance.
When items are goods, agent $i$ is said to \emph{EFX-envy} agent $j$ in allocation $A$
if $v_i(A_j) > v_i(A_i)$ and $v_i(A_i) < v_i(A_j \setminus \{g\})$ for some $g \in A_j$ for which $v_i(g) > 0$.
When items are chores, agent $i$ is said to \emph{EFX-envy} agent $j$ in allocation $A$
if $|v_i(A_j)| < |v_i(A_i)|$ and $|v_i(A_i \setminus \{c\})| > |v_i(A_j)|$ for some $c \in A_i$ for which $v_i(c) < 0$.

Agent $i$ is \emph{EFX}-satisfied in $A$ if $i$ doesn't EFX-envy any other agent.
Allocation $A$ is envy-free up to any good (EFX) if every agent is EFX-satisfied in $A$.
\end{definition}

Despite significant efforts by the fair division community,
the existence of EFX allocations is still an open problem.
A relaxation of EFX, called EF1, has been studied extensively
(in fact, EF1 was proposed before EFX \cite{budish2011combinatorial}).

\begin{definition}[EF1]
Let $\Ical \defeq (N, M, (v_i)_{i \in N})$ be a fair division instance.
When items are goods, agent $i$ is said to \emph{EF1-envy} agent $j$ in allocation $A$
if $v_i(A_i) < v_i(A_j \setminus \{g\})$ for all $g \in A_j$.
When items are chores, agent $i$ is said to \emph{EF1-envy} agent $j$ in allocation $A$
if $|v_i(A_j)| < |v_i(A_i)|$ and $|v_i(A_i \setminus \{c\})| > |v_i(A_j)|$ for all $c \in A_i$.

Agent $i$ is \emph{EFX}-satisfied in $A$ if $i$ doesn't EFX-envy any other agent.
Allocation $A$ is envy-free up to any good (EFX) if every agent is EFX-satisfied in $A$.
\end{definition}

A popular relaxation of proportionality is Maximin Share (MMS) \cite{budish2011combinatorial}.

\begin{definition}[MMS]
Let $\Ical \defeq (N, M, (v_i)_{i \in N})$ be a fair division instance.
Let $\Pi$ be the set of all $|N|$-partitions of $M$.
Define agent $i$'s \emph{maximin share} as
\[ \MMS_i \defeq \max_{X \in \Pi}\min_{j=1}^{|N|} v_i(X_j). \]
Agent $i$ is \emph{MMS-satisfied} by allocation $A$ if $v_i(A_i) \ge \MMS_i$.
Allocation $A$ is \emph{MMS} if every agent is MMS-satisfied by $A$.
\end{definition}

\cite{kurokawa2018fair,feige2022tight} showed, through an intricate counterexample,
that MMS allocations don't always exist. This was a surprising result,
since it contradicted empirical evidence \cite{bouveret2016characterizing}.
Hence, approximations of MMS have been studied.
An allocation $A$ is $\alpha$-MMS, for some $\alpha \in \mathbb{R}$,
if $v_i(A_i) \ge \alpha\cdot\MMS_i$ for every agent $i$.

Recently, \cite{caragiannis2023new,aziz2018knowledge} proposed a notion of fairness
called \emph{Epistemic EFX}, which is a relaxation of EFX.

\begin{definition}[Epistemic EFX]
Agent $i$ is EEFX-satisfied by allocation $A$ if there exists an allocation $B$
(called agent $i$'s EEFX-certificate for $A$) such that $A_i = B_i$ and $i$ is EFX-satisfied by $B$.
An allocation $A$ is Epistemic EFX (EEFX) if every agent is EEFX-satisfied by $A$.
\end{definition}

\subsubsection{Relationships Between Fairness Notions}

The following results, which are trivial to prove, show that
some fairness notions imply other fairness notions.

\begin{observation}
\label[observation]{thm:f-impls}
In a fair division instance,
\begin{enumerate}
\item An EF allocation is also PROP and EFX.
\item $\PROP_i \ge \MMS_i$ for every agent $i$.
\item An EFX allocation is also Epistemic EFX.
\end{enumerate}
\end{observation}

\begin{observation}
\label[observation]{thm:f-impls-n2}
In a fair division instance with only two agents,
\begin{enumerate}
\item an allocation is PROP iff it is EF.
\item an allocation is Epistemic EFX iff it is EFX.
\end{enumerate}
\end{observation}

\subsubsection{Randomized Fairness}
\label{sec:prelims:bobw}

For randomized algorithms, the definition of fairness is more nuanced.
For a fair division instance with two agents and a single good,
no allocation is envy-free or proportional,
but giving the good to a random agent would be fair since
we would be treating both agents equally.

\begin{definition}[ex-ante and ex-post fairness]
Let $\Ical \defeq (N, M, (v_i)_{i \in N})$ be a fair division instance.
Let $A$ be a random allocation (i.e., a random variable supported over allocations).
$A$ is said to be \emph{ex-ante} EF if $\E(v_i(A_i)) \ge \E(v_i(A_j))$ for every $i, j \in N$.
$A$ is said to be \emph{ex-ante} PROP if $\E(v_i(A_i)) \ge v_i(M)/|N|$ for every $i \in N$.
For any fairness notion $F$, $A$ is said to be ex-post $F$ if
every allocation in the support of $A$ is $F$.
\end{definition}

Ex-ante fairness is not enough by itself.
If we had multiple goods and gave all of them to a random agent,
that would be ex-ante EF and ex-ante PROP,
but such an allocation would be very unfair ex-post.
If we had two agents and 5 identical goods,
intuitively, we should first give 2 goods to each agent,
and then allocate the remaining good randomly.
Hence, we aim to get both ex-ante and ex-post fairness simultaneously.

\subsection{The Envy Cycle Elimination Algorithm}
\label{sec:prelims:ece}

Envy Cycle Elimination (ECE) is an algorithm for fair division of goods given by \cite{lipton2004approximately}.
See \cref{algo:ece} for a formal description.
ECE and its variations have been used in many fair division algorithms
\cite{barman2020approximation,amanatidis2020multiple,chaudhury2021little,bhaskar2021approximate,caragiannis2023new,chaudhury2024efx,berger2022almost,feldman2023breakingArxiv}.

\begin{definition}[envy cycle]
Let $A$ be an allocation for a fair division instance.
Let $C \defeq (i_1, \ldots, i_{\ell})$ be a sequence of agents such that
$i_{j-1}$ envies $i_j$ in $A$ for all $j \in [\ell]$, where $i_0 \defeq i_{\ell}$.
Then $C$ is called an \emph{envy cycle} in $A$.

\emph{Resolving} an envy cycle is the operation where we
transfer bundles in the opposite direction of the cycle.
Formally, the operation $\resolveCycle(A, C)$ returns an allocation $B$ where
$B_{i_j} \defeq A_{i_{j+1}}$ for all $j \in [\ell-1]$, $B_{i_{\ell}} \defeq A_{i_1}$,
and $B_k \defeq A_k$ for $k \not\in C$.
\end{definition}

\begin{algorithm}[htb]
\caption{$\ECE(\Ical)$:
The envy cycle elimination algorithm where
$\Ical \defeq ([n], [m], (v_i)_{i=1}^n)$ is a fair division instance for goods.}
\label{algo:ece}
\begin{algorithmic}[1]
\State Let $A \defeq (A_1, \ldots, A_n)$, where $A_i$ is initially $\emptyset$ for all $i \in [n]$.
\For{$j$ from $1$ to $m$}
    \While{there exists an envy cycle $C$ in $A$}\label{alg-line:ece:pick-ec}
        \State $A = \resolveCycle(A, C)$
    \EndWhile
    \State\label{alg-line:ece:pick-unenvied}Find an agent $i$ who is not envied by anyone in $A$.
    \State $A_i = A_i \cup \{j\}$
\EndFor
\State \Return $A$
\end{algorithmic}
\end{algorithm}

Note that there may be multiple cycles at line \ref{alg-line:ece:pick-ec} to choose from,
and multiple unenvied agents at line \ref{alg-line:ece:pick-unenvied} to choose from.
Making these choices randomly can help us randomize ECE.
Unless specified otherwise, we choose arbitrarily.

\subsection{Ordered Instances}
\label{sec:prelims:ordered}

\begin{definition}[ordered instance]
A fair division instance $(N, [m], (v_i)_{i \in N})$ is called \emph{ordered} if
$v_i(1) \ge v_i(2) \ge \ldots \ge v_i(m)$.
\end{definition}

Bouveret and Lema\^itre \shortcite{bouveret2016characterizing} devised a three-step technique to
reduce certain fair division problems to the special case of ordered instances:
\begin{enumerate}
\item Given a fair division instance $\Ical$, convert it to an ordered instance $\Icalhat$.
\item Compute a fair allocation $\Ahat$ for instance $\Icalhat$.
\item Use $\Ahat$ to compute an allocation $A$ for instance $\Ical$.
\end{enumerate}

We call their algorithm for step 1 `$\toOrdered$'
and their algorithm for step 3 `$\pickBySeq$', i.e.,
$\Ical \defeq \toOrdered(\Ical)$ and $A \defeq \pickBySeq(\Ical, \Ahat)$.

For the sake of completeness, we describe $\toOrdered$ and $\pickBySeq$
in \cref{defn:toOrd} and \cref{algo:pickBySeq}, respectively.

\begin{definition}
\label[definition]{defn:toOrd}
For the fair division instance $\Ical \defeq (N, M, (v_i)_{i \in N})$,
$\toOrdered(\Ical)$ is defined as the instance $(N, [|M|], (\vhat_i)_{i \in N})$, where
for each $i \in N$ and $j \in [|M|]$, $\vhat_i(j)$ is
the $j\Th$ largest number in the multiset $\{v_i(g) \mid g \in M\}$.
\end{definition}

\begin{algorithm}[htb]
\caption{$\pickBySeq(\Ical, \Ahat)$:
$\Ical \defeq (N, M, (v_i)_{i \in N})$ is a fair division instance.
$\Ahat$ is an allocation for instance $\toOrdered(\Ical)$.
}
\label{algo:pickBySeq}
\begin{algorithmic}[1]
\State For all $j \in [|M|]$, let $a_j$ be the agent who gets item $j$ in $\Ahat$.
\State Let $A \defeq (A_1, \ldots, A_n)$, where $A_i$ is initially $\emptyset$ for all $i \in [n]$.
\State Set $S = M$.
\For{$t$ from $1$ to $|M|$}
    \State Let $j = \argmax_{j \in S} v_{a_t}(j)$.
        \Comment{$j$ is $a_t$'s favorite item in $S$.}
    \State $A_{a_t} = A_{a_t} \cup \{j\}$. \quad $S = S \setminus \{j\}$.
        \Comment{Allocate $j$ to $a_t$.}
\EndFor
\State \Return $A$
\end{algorithmic}
\end{algorithm}

\begin{lemma}
\label[lemma]{thm:toOrd-time}
$\toOrdered$ and $\pickBySeq$ run in $O(mn\log m)$ time,
where $n$ is the number of agents and $m$ is the number of goods.
\end{lemma}
\begin{proof}
$\toOrdered$:
For each agent $i \in [n]$, we can compute $(\vhat_i(j))_{j \in [m]}$ in $O(m \log m)$ time
by sorting the values $\{v_i(g) \mid g \in [m]\}$.

$\pickBySeq$:
Each agent $i \in [n]$ maintains a balanced binary search tree $T_i$.
At any point in time, let $S$ be the set of remaining (unallocated) items.
Then $T_i$ contains $\{(v_i(j), j): j \in S\}$, where pairs are compared lexicographically.

Initially, $S = [m]$, so $T_i$ can be initialized in $O(m \log m)$ time for each $i \in [n]$.
In each round, if it's agent $\ihat$'s turn to pick,
she finds the largest element $(v_{\ihat}(\jhat), \jhat)$ in $T_{\ihat}$ in $O(\log m)$ time.
Then for each agent $i$, we update $T_i$ by deleting $(v_i(\jhat), \jhat)$ in $O(\log m)$ time.
Hence, the time in each round is $O(n\log m)$. Since there are $m$ rounds,
the total running time of $\pickBySeq$ is $O(mn\log m)$.
\end{proof}

The following result shows that if the allocation $\Ahat$ is fair,
then allocation $A$ is also fair.

\begin{lemma}
\label[lemma]{thm:ord-inc-self}
\label[lemma]{thm:efx-ord-eefx}
Let $\Ical \defeq (N, M, (v_i)_{i \in N})$ and
$\Icalhat \defeq \toOrdered(\Ical) = (N, [|M|], (\vhat_i)_{i \in N})$.
Let $\Ahat$ be an allocation for $\Icalhat$
and $A \defeq \pickBySeq(\Ical, \Ahat)$. Then
\begin{enumerate}
\item $v_i(A_i) \ge \vhat_i(\Ahat_i)$ for all $i \in N$.
\item If $\Ahat$ is EFX, then $A$ is epistemic EFX.
\item For all $i \in N$, $i$'s maximin share is the same in $\Ical$ and $\Icalhat$.
\end{enumerate}
\end{lemma}
\begin{proof}
Claims 1 and 2 follow from Lemmas 2 and 3 of \cite{caragiannis2023new} for goods.
Claim 3 is true because $\toOrdered$ only permutes the values of the goods.
The set of values remains the same.
The proof for chores is similar.
\end{proof}

\section{Efficiently Dividing Goods Among Two Agents}
\label{sec:n2algo}

We give a randomized algorithm for the fair division of goods among two agents
whose output is ex-post EFX and ex-ante EF.
The algorithm runs in $O(m \log m)$ time, where $m$ is the number of goods.
The algorithm can be viewed as a randomization of Envy-Cycle-Elimination.

\subsection{Ordered Instances}

We first give an algorithm called $\ECEGTwo$ (\cref{algo:eceGN2}) for ordered instances.
(We show how to extend it to arbitrary instances later.)
We maintain two allocations, $A$ and $B$, which are both initially empty.
The algorithm proceeds in rounds. In each round, we first resolve envy cycles in both $A$ and $B$.
Then we find an unenvied agent $i_A$ in $A$ and an unenvied agent $i_B$ in $B$
such that $i_A \neq i_B$.
We then allocate a good to $i_A$ in $A$ and the same good to $i_B$ in $B$.

\begin{algorithm}[htb]
\caption{$\ECEGTwo(\Ical)$: Algorithm for dividing goods among two agents.
\\ \textbf{Input}: A fair division instance $\Ical \defeq ([2], [m], (v_1, v_2))$.
\\ \textbf{Output}: A pair of allocations.}
\label{algo:eceGN2}
\begin{algorithmic}[1]
\State $A_1 = A_2 = B_1 = B_2 = \emptyset$.
\For{$g$ from $1$ to $m$}\label{alg-line:eceGN2:oloop}
    \For{$C$ in $\{A, B\}$}
        \State $S_C = \{i \in [2]: v_{3-i}(C_{3-i}) \ge v_{3-i}(C_i)\}$.
            \Comment{Set of unenvied agents in allocation $C$.}
    \EndFor
    \If{$1 \in S_A$ and $2 \in S_B$}
        \State $A_1 = A_1 \cup \{g\}$.\quad $B_2 = B_2 \cup \{g\}$.
    \ElsIf{$2 \in S_A$ and $1 \in S_B$}
        \State $A_2 = A_2 \cup \{g\}$.\quad $B_1 = B_1 \cup \{g\}$.
    \Else
        \State \label{alg-line:eceGN2:error} \textbf{error}
    \EndIf\label{alg-line:eceGN2:postAlloc}
    \For{$C$ in $\{A, B\}$}
        \If{$v_1(C_2) > v_1(C_1)$ and $v_2(C_1) > v_2(C_2)$}
            \Comment{There is an envy cycle in $C$}
            \State \label{alg-line:eceGN2:ece}Swap $C_1$ and $C_2$.
                \Comment{Resolve the envy cycle.}
        \EndIf
    \EndFor
\EndFor
\State \Return $((A_1, A_2), (B_1, B_2))$.
\end{algorithmic}
\end{algorithm}

\begin{lemma}
\label[lemma]{thm:eceGN2-time}
$\ECEGTwo$ (\cref{algo:eceGN2}) runs in $O(m)$ time, where $m$ is the number of goods.
\end{lemma}
\begin{proof}
Throughout the algorithm, we maintain $v_i(C_j)$ for each $i \in [2]$, $j \in [2]$, $C \in \{A, B\}$.
We can update these in $O(1)$ time whenever any good is added to $C_j$.
There are $m$ rounds, and each round can be implemented in $O(1)$ time.
\end{proof}

We show that the average value any agent $i$ gets from $\ECEGTwo$'s output
is at least her proportional share.

\begin{lemma}
\label[lemma]{thm:eceGN2-prop}
For any fair division instance $\Ical \defeq ([2], [m], (v_1, v_2))$,
$\ECEGTwo(\Ical)$ never throws an error (line \ref{alg-line:eceGN2:error}),
and its output $(A, B)$ satisfies $v_i(A_i) + v_i(B_i) \ge v_i([m])$ for all $i \in [2]$.
\end{lemma}
\begin{proof}
For $0 \le t \le m$, define the predicate $P(t)$ as:
\\ \textsl{$\ECEGTwo(\Ical)$ never throws an error in the first $t$ rounds
    and satisfies $v_i(A_i) + v_i(B_i) \ge v_i([t])$
    for all $i \in [2]$ after the first $t$ rounds.}

Note that the condition $v_i(A_i) + v_i(B_i) \ge v_i([t])$
is equivalent to $v_i(A_i) + v_i(B_i) \ge v_i(A_{3-i}) + v_i(B_{3-i})$,
since $A_i \cup A_{3-i} = B_i \cup B_{3-i} = [t]$ immediately after $t$ rounds.

We will prove $P(t)$ for all $t$ using induction.
$P(0)$ is trivially true. Now assume $P(t-1)$ for some $t \in [m]$.
Let $(A^{(1)}, B^{(1)})$ be the value of $(A, B)$ at the beginning of the $t\Th$ iteration.

Suppose $\ECEGTwo$ throws an error in the $t\Th$ iteration.
Then $S_A = S_B = \{i\}$ for some $i \in [2]$ in that iteration. Let $j \defeq 3-i$.
Then $i$ envies $j$ in both $A^{(1)}$ and $B^{(1)}$, i.e.,
$v_i(A^{(1)}_i) < v_i(A^{(1)}_j)$ and $v_i(B^{(1)}_i) < v_i(B^{(1)}_j)$.
But this contradicts $P(t-1)$. Hence, $\ECEGTwo$ doesn't throw an error in the $t\Th$ iteration.

In the $t\Th$ round, let $(A^{(2)}, B^{(2)})$ be the value of $(A, B)$ after line \ref{alg-line:eceGN2:postAlloc}
and let $(A^{(3)}, B^{(3)})$ be the value of $(A, B)$ at the end of the round.
Suppose agent $i$ gets good $t$ in $A$ and agent $j$ gets good $t$ in $B$, where $i \neq j$.
Then $v_i(A^{(2)}_i) + v_i(B^{(2)}_i) = v_i(A^{(1)}_i) + v_i(B^{(1)}_i) + v_i(t) \ge v_i([t])$.
Moreover, envy cycle elimination (line \ref{alg-line:eceGN2:ece}) increases each agent's value
for her own bundle, so $v_i(C^{(3)}_i) \ge v_i(C^{(2)}_i)$ for all $C \in \{A, B\}$.
Hence, $v_i(A^{(3)}_i) + v_i(B^{(3)}_i) \ge v_i([t])$.
Hence, $P(t)$ holds.
By mathematical induction, we get that $P(m)$ is true.
\end{proof}

\begin{lemma}
\label[lemma]{thm:eceGN2-efx}
Let $(A, B) = \ECEGTwo(\Ical)$, where $\Ical$ is an ordered fair division instance.
Then $A$ and $B$ are EFX and $4/5$-MMS.
\end{lemma}
\begin{proof}
Lemma 3.5 in \cite{barman2020approximation} says that the output of
the Envy-Cycle-Elimination algorithm on an ordered instance is EFX.
Since $\ECEGTwo$ is equivalent to running two instances of the Envy-Cycle-Elimination algorithm
in parallel, $A$ and $B$ are EFX.
$A$ and $B$ are $4/5$-MMS by Theorem 3.1 in \cite{barman2020approximation}.
\end{proof}

Hence, if we pick one of the allocations output by $\ECEGTwo$ on an ordered instance at random,
the output is ex-post EFX (by \cref{thm:eceGN2-efx}) and ex-ante EF (by \cref{thm:eceGN2-prop}).
(Recall that when there are only two agents, an allocation is EF iff it is PROP.)

\subsection{Extending Beyond Ordered Instances}

We obtain \cref{algo:eceGN2b} by combining algorithm $\ECEGTwo$
with the technique in \cref{sec:prelims:ordered}.

\begin{algorithm}[htb]
\caption{Randomized algorithm for dividing goods among two agents.
Takes a fair division instance $\Ical$ as input.}
\label{algo:eceGN2b}
\begin{algorithmic}[1]
\State $\Icalhat = \toOrdered(\Ical)$.
\State $(\Ahat, \Bhat) = \ECEGTwo(\Icalhat)$.
\State $A = \pickBySeq(\Ical, \Ahat)$.
\State $B = \pickBySeq(\Ical, \Bhat)$.
\State \Return an allocation uniformly at random from $\{A, B\}$.
\end{algorithmic}
\end{algorithm}

\begin{theorem}
\label[theorem]{thm:eceGN2b}
The output of \cref{algo:eceGN2b} is ex-post EFX, ex-post $4/5$-MMS, and ex-ante EF.
\end{theorem}
\begin{proof}
Let $\Ical \defeq ([2], [m], (v_1, v_2))$ be the input to \cref{algo:eceGN2b}.
Let $\Icalhat \defeq ([2], [m], (\vhat_1, \vhat_2))$.
$\Ahat$ and $\Bhat$ are EFX and $4/5$-MMS for $\Icalhat$ by \cref{thm:eceGN2-efx}.
By \cref{thm:efx-ord-eefx,thm:f-impls-n2}, $A$ and $B$ are also EFX and $4/5$-MMS.
By \cref{thm:ord-inc-self,thm:eceGN2-prop},
$v_i(A_i) + v_i(B_i) \ge \vhat_i(\Ahat_i) + \vhat_i(\Bhat_i) \ge v_i([m])$ for all $i \in [2]$.
This implies $v_i(A_i) + v_i(B_i) \ge v_i(A_{3-i}) + v_i(B_{3-i})$ for all $i \in [2]$.
Hence, picking uniformly randomly from $\{A, B\}$ is ex-post EFX and ex-ante EF.
\end{proof}

\begin{lemma}
\Cref{algo:eceGN2b} runs in $O(m\log m)$ time.
\end{lemma}
\begin{proof}
Follows from \cref{thm:eceGN2-time,thm:toOrd-time}.
\end{proof}

\subsection{Chores}

Our results for two agents can easily be adapted to the chores setting.
We first modify $\ECEGTwo$ (\cref{algo:eceGN2}) to get $\ECECTwo$ (\cref{algo:eceCN2}).

\begin{algorithm}[htb]
\caption{$\ECECTwo(\Ical)$: Algorithm for dividing chores among two agents.
\\ \textbf{Input}: A fair division instance $\Ical \defeq ([2], [m], (v_1, v_2))$.
\\ \textbf{Output}: A pair of allocations.}
\label{algo:eceCN2}
\begin{algorithmic}[1]
\State $A_1 = A_2 = B_1 = B_2 = \emptyset$.
\For{$c$ from $m$ to $1$}\label{alg-line:eceCN2:oloop}
    \For{$C$ in $\{A, B\}$}
        \State $S_C = \{i \in [2]: |v_i(C_i)| \le |v_i(C_{3-i})|\}$.
            \Comment{Set of envy-free agents in allocation $C$.}
    \EndFor
    \If{$1 \in S_A$ and $2 \in S_B$}
        \State $A_1 = A_1 \cup \{c\}$.\quad $B_2 = B_2 \cup \{c\}$.
    \ElsIf{$2 \in S_A$ and $1 \in S_B$}
        \State $A_2 = A_2 \cup \{c\}$.\quad $B_1 = B_1 \cup \{c\}$.
    \Else
        \State \label{alg-line:eceCN2:error} \textbf{error}
    \EndIf\label{alg-line:eceCN2:postAlloc}
    \For{$C$ in $\{A, B\}$}
        \If{$|v_1(C_2)| < |v_1(C_1)|$ and $|v_2(C_1)| < |v_2(C_2)|$}
            \Comment{$\exists$ an envy cycle in $C$}
            \State \label{alg-line:eceCN2:ece}Swap $C_1$ and $C_2$.
                \Comment{Resolve the envy cycle.}
        \EndIf
    \EndFor
\EndFor
\State \Return $((A_1, A_2), (B_1, B_2))$.
\end{algorithmic}
\end{algorithm}

\begin{lemma}
\label[lemma]{thm:eceCN2-time}
$\ECECTwo$ (\cref{algo:eceCN2}) runs in $O(m)$ time, where $m$ is the number of chores.
\end{lemma}
\begin{proof}
Throughout the algorithm, we maintain $v_i(C_j)$ for each $i \in [2]$, $j \in [2]$, $C \in \{A, B\}$.
We can update these in $O(1)$ time whenever any chore is added to $C_j$.
There are $m$ rounds, and each round can be implemented in $O(1)$ time.
\end{proof}

\begin{lemma}
\label[lemma]{thm:eceCN2-prop}
For any fair division instance $\Ical \defeq ([2], [m], (v_1, v_2))$,
$\ECECTwo(\Ical)$ never throws an error (line \ref{alg-line:eceCN2:error}),
and its output $(A, B)$ satisfies $|v_i(A_i)| + |v_i(B_i)| \le |v_i([m])|$ for all $i \in [2]$.
\end{lemma}
\begin{proof}
For $0 \le t \le m$, define the predicate $P(t)$ as:
\\ \textsl{$\ECECTwo(\Ical)$ never throws an error in the first $t$ rounds
    and satisfies $|v_i(A_i)| + |v_i(B_i)| \le |v_i([m] \setminus [m-t])|$
    for all $i \in [2]$ after the first $t$ rounds.}

Note that the condition $|v_i(A_i)| + |v_i(B_i)| \le |v_i([m] \setminus [m-t])|$
is equivalent to $|v_i(A_i)| + |v_i(B_i)| \le |v_i(A_{3-i})| + |v_i(B_{3-i})|$,
since $A_i \cup A_{3-i} = B_i \cup B_{3-i} = [m] \setminus [m-t]$ immediately after $t$ rounds.

We will prove $P(t)$ for all $t$ using induction.
$P(0)$ is trivially true. Now assume $P(t-1)$ for some $t \in [m]$.
Let $(A^{(1)}, B^{(1)})$ be the value of $(A, B)$ at the beginning of the $t\Th$ iteration.

Suppose $\ECECTwo$ throws an error in the $t\Th$ iteration.
Then $S_A = S_B = \{i\}$ for some $i \in [2]$ in that iteration. Let $j \defeq 3-i$.
Then $j$ envies $i$ in both $A^{(1)}$ and $B^{(1)}$, i.e.,
$|v_j(A^{(1)}_i)| < |v_j(A^{(1)}_j)|$ and $|v_j(B^{(1)}_i)| < |v_j(B^{(1)}_j)|$.
But this contradicts $P(t-1)$. Hence, $\ECECTwo$ doesn't throw an error in the $t\Th$ iteration.

In the $t\Th$ round, let $(A^{(2)}, B^{(2)})$ be the value of $(A, B)$ after line \ref{alg-line:eceCN2:postAlloc}
and let $(A^{(3)}, B^{(3)})$ be the value of $(A, B)$ at the end of the round.
Suppose agent $i$ gets chore $m-t+1$ in $A$ and agent $j$ gets chore $m-t+1$ in $B$, where $i \neq j$.
Then $|v_i(A^{(2)}_i)| + |v_i(B^{(2)}_i)| = |v_i(A^{(1)}_i)| + |v_i(B^{(1)}_i)| + |v_i(m-t+1)| \ge |v_i([m] \setminus [m-t])|$.
Moreover, envy cycle elimination (line \ref{alg-line:eceCN2:ece}) decreases each agent's disutility
for her own bundle, so $|v_i(C^{(3)}_i)| \le |v_i(C^{(2)}_i)|$ for all $C \in \{A, B\}$.
Hence, $|v_i(A^{(3)}_i)| + |v_i(B^{(3)}_i)| \ge |v_i([m] \setminus [m-t])|$.
Hence, $P(t)$ holds.
By mathematical induction, we get that $P(m)$ is true.
\end{proof}

\begin{lemma}
\label[lemma]{thm:eceCN2-efx}
Let $(A, B) = \ECECTwo(\Ical)$, where $\Ical$ is an ordered fair division instance.
Then $A$ and $B$ are EFX and $7/6$-MMS.
\end{lemma}
\begin{proof}
The proof for EFX is similar to \cref{thm:eceGN2-efx}.
$A$ and $B$ are $7/6$-MMS by Theorem A.2 of \cite{barman2020approximation}.
\end{proof}

Hence, if we pick one of the allocations output by $\ECECTwo$ on an ordered instance at random,
the output is ex-post EFX (by \cref{thm:eceCN2-efx}) and ex-ante EF (by \cref{thm:eceCN2-prop}).
(Recall that when there are only two agents, an allocation is EF iff it is PROP.)

We can get an algorithm for non-ordered instances (\cref{algo:eceCN2b}) by
combining algorithm $\ECECTwo$ with the technique in \cref{sec:prelims:ordered}.

\begin{algorithm}[htb]
\caption{Randomized algorithm for dividing chores among two agents.
Takes a fair division instance $\Ical$ as input.}
\label{algo:eceCN2b}
\begin{algorithmic}[1]
\State $\Icalhat = \toOrdered(\Ical)$.
\State $(\Ahat, \Bhat) = \ECECTwo(\Icalhat)$.
\State $A = \pickBySeq(\Ical, \Ahat)$.
\State $B = \pickBySeq(\Ical, \Bhat)$.
\State \Return an allocation uniformly at random from $\{A, B\}$.
\end{algorithmic}
\end{algorithm}

\begin{theorem}
\label[theorem]{thm:eceCN2b}
The output of \cref{algo:eceCN2b} is ex-post EFX, ex-post $7/6$-MMS, and ex-ante EF.
\end{theorem}
\begin{proof}
Let $\Ical \defeq ([2], [m], (v_1, v_2))$ be the input to \cref{algo:eceCN2b}.
Let $\Icalhat \defeq ([2], [m], (\vhat_1, \vhat_2))$.
$\Ahat$ and $\Bhat$ are EFX and $7/6$-MMS for $\Icalhat$ by \cref{thm:eceCN2-efx}.
By \cref{thm:efx-ord-eefx,thm:f-impls-n2}, $A$ and $B$ are also EFX and $7/6$-MMS.
By \cref{thm:ord-inc-self,thm:eceCN2-prop},
$|v_i(A_i)| + |v_i(B_i)| \le |\vhat_i(\Ahat_i)| + |\vhat_i(\Bhat_i)| \le |v_i([m])|$ for all $i \in [2]$.
This implies $|v_i(A_i)| + |v_i(B_i)| \le |v_i(A_{3-i})| + |v_i(B_{3-i})|$ for all $i \in [2]$.
Hence, picking uniformly randomly from $\{A, B\}$ is ex-post EFX and ex-ante EF.
\end{proof}

\begin{lemma}
\Cref{algo:eceCN2b} runs in $O(m\log m)$ time.
\end{lemma}
\begin{proof}
Follows from \cref{thm:eceCN2-time,thm:toOrd-time}.
\end{proof}

\section{Hard Input for Three Agents}
\label{sec:ce}

In this section, we describe a family of ordered fair division instances
with 3 agents and 9 goods for which many natural ways of randomizing the
Envy-Cycle-Elimination (ECE) algorithm fail to give ex-ante PROP.

In this section, we use \cref{algo:ece} as our definition of the ECE algorithm.
(There are some variants of \cref{algo:ece}, like the one in \cite{feldman2023breakingArxiv},
that resolve cycles only until an unenvied agent appears, instead of resolving all envy cycles.
We do not study such variants in this work.)

To create a difficult instance for ECE, we start with 3 agents and 9 goods of values
3, 3, 3, 2, 1, 1, 1, 1, 1. We then carefully perturb these values to ensure that
ECE consistently favors one of the agents. The instance is given by \cref{ex:goods-ce}.
Note that good 4 is the only good for which agents have different values.

\begin{table}[htb]
\centering
\caption{Example instance for $n=3$.
Here $\delta \in (0, 1/4]$ and $\eps \in (0, \delta/3]$.}
\label{ex:goods-ce}
\begin{tabular}{cccccccccc}
\toprule $g$ & 1 & 2 & 3 & 4 & 5 & 6 & 7 & 8 & 9
\\ \midrule $v_1(g)$ & $3+2\eps$ & $3+\eps$ & $3$ & $2-\delta$  & $1+\delta$ & 1 & 1 & 1 & 1
\\     $v_2(g), v_3(g)$ & $3+2\eps$ & $3+\eps$ & $3$ & $2+2\delta$ & $1+\delta$ & 1 & 1 & 1 & 1
\\ \bottomrule
\end{tabular}
\end{table}

For the instance given by \cref{ex:goods-ce}, after allocating the first 2 goods,
one can verify that the remaining algorithm is deterministic
(because at any point in the algorithm, there is at most one envy cycle
and the unenvied agent is unique).
The possible outputs of ECE are listed in \cref{ex:goods-ece-out}.

\begin{table}[htb]
\centering
\caption{Outputs of ECE for \cref{ex:goods-ce}.}
\label{ex:goods-ece-out}
\begin{tabular}{ccc}
\toprule Good 1's owner & Good 2's owner & Allocation
\\ \midrule
   1 & 3 & $A^{(1)} = \{\{1,6,7,9\},\{3,4\},\{2,5,8\}\}$
\\ 2 & 3 & $A^{(2)} = \{\{1,6,7,9\},\{3,4\},\{2,5,8\}\}$ %
\\ 1 & 2 & $A^{(3)} = \{\{1,6,7,9\},\{2,5,8\},\{3,4\}\}$
\\ 3 & 2 & $A^{(4)} = \{\{1,6,7,9\},\{2,5,8\},\{3,4\}\}$ %
\\ 2 & 1 & $A^{(5)} = \{\{2,5,8\},\{1,6,7,9\},\{3,4\}\}$
\\ 3 & 1 & $A^{(6)} = \{\{2,5,8\},\{3,4\},\{1,6,7,9\}\}$
\\ \bottomrule
\end{tabular}
\end{table}

Suppose we randomize over the allocations in \cref{ex:goods-ece-out}
by picking $A^{(k)}$ with probability $p_k$ for each $k \in [6]$.
We show (in \cref{thm:goods-ce-p}) that for the output distribution to be ex-ante PROP,
$p_5$ and $p_6$ must each be approximately $1/3$ for small $\delta$.

Known randomizations of ECE (\cite{feldman2023breakingArxiv}, and Appendix D of \cite{babaioff2022bestArxiv})
first randomly assign the first $n$ goods, without considering the remaining instance,
and then resume running (an appropriately randomized) ECE algorithm on the remaining goods.
We call such algorithms \emph{two-phase-online}.

We can show that two-phase-online randomizations of ECE cannot give us ex-ante PROP.
Consider the instance in \cref{ex:goods-ce} and the two other instances
obtainable by permuting the agents. Any two-phase-online randomization of ECE
will use the same value of $p$ for these three instances, because $p$ depends on
only the first 3 goods, and the agents have identical valuations for them.
By applying \cref{thm:goods-ce-p} to each instance, we get that
$p_j$ should be approximately $1/3$ for each $j \in [6]$, which is impossible.
Hence, no two-phase-online randomization of ECE can give us ex-ante PROP.

\begin{lemma}
\label[lemma]{thm:goods-ce-p}
For the fair division instance given in \cref{ex:goods-ce}, suppose we output
allocation $A^{(k)}$ with probability $p_k$ for each $k \in [6]$ (c.f.~\cref{ex:goods-ece-out}).
If the output distribution is ex ante PROP, then for all $j \in \{5, 6\}$, we have
\[ \frac{1-3(\delta-\eps)}{1-2(\delta-\eps)} \le 3p_j \le \frac{1+3\delta}{1-2\delta+2\eps}. \]
\end{lemma}
\begin{proof}
\Cref{table:goods-ce-bv} lists the values agents have for their own bundle
for each allocation in \cref{ex:goods-ece-out}.

\begin{table}[htb]
\centering
\caption{Agents' values for their own bundles in \cref{ex:goods-ece-out}.}
\label{table:goods-ce-bv}
\begin{tabular}{ccccccc}
\toprule $i$ & $v_i(A^{(1)}_i)$ & $v_i(A^{(2)}_i)$ & $v_i(A^{(3)}_i)$ & $v_i(A^{(4)}_i)$ & $v_i(A^{(5)}_i)$ & $v_i(A^{(6)}_i)$
\\ \midrule
   1 & $\mathbf{6}+2\eps$ & $\mathbf{6}+2\eps$ & $\mathbf{6}+2\eps$ & $\mathbf{6}+2\eps$ & $5+\delta+\eps$ & $5+\delta+\eps$
\\ 2 & $5+2\delta$ & $5+2\delta$ & $5+\delta+\eps$ & $5+\delta+\eps$ & $\mathbf{6}+2\eps$ & $5+2\delta$
\\ 3 & $5+\delta+\eps$ & $5+\delta+\eps$ & $5+2\delta$ & $5+2\delta$ & $5+2\delta$ & $\mathbf{6}+2\eps$
\\ \bottomrule
\end{tabular}
\end{table}

Note that $\PROP_1 = 5 + 1/3 + \eps$ and $\PROP_2 = \PROP_3 = 5 + 1/3 + \delta + \eps$.
Let $y_i$ be agent $i$'s expected utility from their own bundle, i.e.,
$y_i \defeq \sum_{j=1}^6 v_i(A^{(j)}_i)p_j$.
Suppose the output distribution is ex ante PROP.
Then $y_i \ge \PROP_i$ for all $i \in [3]$. Hence,
\begin{align*}
& 5 + 1/3 + \eps = \PROP_1 \le y_1 = (6+2\eps)(1-p_5-p_6) + (5+\delta+\eps)(p_5 + p_6)
\\ & \implies p_5 + p_6 \le \frac{2/3+\eps}{1-\delta+\eps} \le \frac{2/3+\eps}{1-2\delta+2\eps}.
\\ & 5 + 1/3 + \delta + \eps = \PROP_2 \le y_2 \le (6+2\eps)p_5 + (5+2\delta)(1-p_5)
\\ & \implies p_5 \ge \frac{1/3-\delta+\eps}{1-2\delta+2\eps}.
\\ & 5 + 1/3 + \delta + \eps = \PROP_3 \le y_3 \le (6+2\eps)p_6 + (5+2\delta)(1-p_6)
\\ & \implies p_6 \ge \frac{1/3-\delta+\eps}{1-2\delta+2\eps}.
\end{align*}
\[ p_5 \le \frac{2/3+\eps}{1-2\delta+2\eps} - p_6 \le \frac{1/3+\delta}{1-2\delta+2\eps}. \]
Similarly,
\[ p_6 \le \frac{1/3+\delta}{1-2\delta+2\eps}. \qedhere \]
\end{proof}

\section{Conclusion}
\label{sec:conclusion}

We investigated whether ECE can be randomized to obtain better fairness guarantees
than the state-of-the-art.
For two agents, we show that such a randomization is possible
and gives us essentially the best we can hope for: ex-ante EF and ex-post EFX in $O(m \log m)$ time.
We extend this result to chores too.

However, for more than two agents, ECE doesn't seem to be a very promising direction to pursue,
since natural randomizations of ECE fail to give ex-ante PROP.
The version of ECE that we consider, though, is one where we resolve all envy cycles
before allocating a good to an unenvied agent. This is unlike \citet{feldman2023breakingArxiv},
where they sometimes allocate a good to an unenvied agent even if envy cycles are present.
Perhaps one can circumvent our negative results in \cref{sec:ce}
by avoiding such an aggressive decycling.

Since the existence of EFX alloctions has been a difficult open problem for a long time,
it is reasonable to first focus on its relaxations for ex-post fairness.
The best relaxation that we know of (in our subjective opinion) is Epistemic EFX.
For ex-ante fairness, EF and PROP are the dominant notions.
We believe that PROP is not much worse than EF,
for the same reason that Epistemic EFX is not much worse than EFX
(see \cite{caragiannis2023new} for discussion regarding this).
Hence, obtaining ex-ante PROP and ex-post Epistemic EFX is,
in our opinion, a key open problem in the field of best-of-both-worlds fair division.

\end{document}